\documentclass[11pt]{article}

\newcommand{\ignore}[1]{}%

\usepackage{algorithm, algorithmic}
\usepackage{amsmath}
\usepackage{amssymb}
\usepackage{amsthm}
\usepackage{tabularx}
\usepackage{environ}
\usepackage{graphicx}
\usepackage{subcaption}
\usepackage{url}
\usepackage[margin=1in]{geometry}
\usepackage{todonotes}
\usepackage{tikz}
\usetikzlibrary{arrows,chains,shapes,matrix,positioning,calc}

\newcommand{\INPUT}{\item[{\bf Input:}]}
\renewcommand{\algorithmiccomment}[1]{\bgroup\hfill\footnotesize~#1\egroup}

\usepackage[framemethod=tikz]{mdframed}
\mdfsetup{
	roundcorner=4pt,
	nobreak=true,
	linewidth=.5,
	innerleftmargin=25pt,
	innerrightmargin=25pt,
	innertopmargin=15pt,
	innerbottommargin=15pt,
	skipabove=12,skipbelow=8pt
}

\tikzset{
block/.style={
    rectangle,
    draw,
    text width=13em,
    text centered,
    rounded corners
},
rectangle connector/.style={
    connector,
    to path={(\tikztostart) -- ++(#1,0pt) \tikztonodes |- (\tikztotarget) },
    pos=0.5
},
rectangle connector/.default=-2cm,
straight connector/.style={
    connector,
    to path=--(\tikztotarget) \tikztonodes
}
}

\let\emptyset\varnothing

\makeatletter
\newcommand{\problemtitle}[1]{\gdef\@problemtitle{#1}}%
\newcommand{\probleminput}[1]{\gdef\@probleminput{#1}}%
\newcommand{\problemoutput}[1]{\gdef\@problemoutput{#1}}%
\NewEnviron{problem}{
  \problemtitle{}\probleminput{}\problemoutput{}%
  \BODY%
  \par\addvspace{.5\baselineskip}
  \noindent
  \begin{tabularx}{\textwidth}{@{\hspace{\parindent}} l X c}
    \multicolumn{2}{@{\hspace{\parindent}}l}{\@problemtitle} \\
    \textbf{Input:} & \@probleminput \\
    \textbf{Output:} & \@problemoutput%
  \end{tabularx}
  \par\addvspace{.5\baselineskip}
}

\theoremstyle{plain}
\newtheorem{theorem}{Theorem}
\newtheorem*{theorem*}{Theorem}
\newtheorem{lemma}[theorem]{Lemma}
\newtheorem{claim}[theorem]{Claim}
\newtheorem{proposition}[theorem]{Proposition}

\newtheorem{corollary}[theorem]{Corollary}

\theoremstyle{definition}
\newtheorem{definition}{Definition}

\theoremstyle{remark}
\newtheorem*{remark}{Remark}

\newcommand{\poly}{\operatorname{poly}}
\newcommand{\ED}{\operatorname{ED}}

\newcommand{\eps}{\epsilon}

\newcommand{\Ham}{\operatorname{Ham}}

\newcommand{\method}{\operatorname}
\newcommand{\OUTPUT}{\\\hspace{-0.15in}\textbf{Output:} }

\newcommand{\MaxAlign}{\method{MaxAlign}}
\newcommand{\MaxShiftAlign}{\method{MaxShiftAlign}}
\newcommand{\True}{{\textbf{True}}}

\newcommand{\tO}{{\tilde{O}}}

\title{A Simple Sublinear Algorithm for Gap Edit Distance}
\author{Joshua Brakensiek\thanks{Stanford University, partially supported by an NSF Graduate Research Fellowship.} \and Moses Charikar\thanks{Stanford University, supported by a Simons Investigator Award, a Google Faculty Research Award and an Amazon Research Award.}\and Aviad Rubinstein\thanks{Stanford University.}}
\date{}

\begin{document}

\maketitle
\begin{abstract}

We study the problem of estimating the edit distance between two $n$-character strings.
While exact computation in the worst case is believed to require near-quadratic time, previous work showed that in certain regimes it is possible to solve the following {\em gap edit distance} problem in {\em sub-linear time}: distinguish between inputs of distance $\le k$ and $>k^2$. 
Our main result is a very simple algorithm for this benchmark that runs in time $\tO(n/\sqrt{k})$, and in particular settles the open problem of obtaining a truly sublinear time for the entire range of relevant $k$. %

Building on the same framework, we also obtain a $k$-vs-$k^2$ algorithm for the one-sided preprocessing model with $\tO(n)$ preprocessing time and $\tO(n/k)$ query time (improving over a recent $\tO(n/k+k^2)$-query time algorithm for the same problem~\cite{GRS20}).

 \end{abstract}

\setcounter{page}{0}
\thispagestyle{empty}
\pagebreak

\section{Introduction}\label{sec:intro}

We study the problem of estimating the edit distance between two $n$-character strings.
There is a classic $O(n^2)$ dynamic programming algorithm, and fine-grained complexity results from recent years suggest that it is nearly optimal~\cite{BK15-LCS,BI18,AHWW16-polylog_shaved,AB18}.
There have been long lines of works on beating the quadratic time barrier with approximations~\cite{BEKMRRS03-edit-testing,BJKK04-edit,BES06-edit,AO12-edit,AKO10-edit,BEGHS18-edit-quantum,CDGKS18,CGKK18-ED,HSSS19,RSSS19,BR20, KS20,GRS20,RS20,AN20}, or beyond-worst case~\cite{Ukkonen85,Apostolico86,Meyers86,LMS98,AK12,ABBK17,BK18,Kuszmaul19,HRS19,BSS20}. 
Motivated by applications where the strings may be extremely long (e.g.~bioinformatics), we are interested in algorithms that run even faster, namely in  sub-{\em linear} time. For exact computation in the worst case, this is unconditionally impossible --- even distinguishing between a pair of identical strings and a pair that differs in a single character requires reading the entire input. 
But in many regimes sublinear algorithms are still possible~\cite{BEKMRRS03-edit-testing,BJKK04-edit,CK06,AN10, AO12-edit,SS17, GKS19,NRRS19,BCLW19,RSSS19}.

\subsubsection*{Gap Edit Distance: $k$ vs $k^2$}
We give new approximation algorithms for edit distance that run in sublinear time when the input strings are close. 
To best understand our contribution and how it relates to previous work, we focus on the benchmark advocated by~\cite{GKS19} of distinguishing input strings whose edit distance is $\le k$ from $\gtrsim k^2$; we discuss more general parameters later in this section. 
Notice that we can assume wlog that $k < \sqrt{n}$ (otherwise the algorithm can always accept). Furthermore, for tiny $k$ there is an unconditional easy lower bound of $\Omega(n/k^2)$ for distinguishing even identical strings from ones with $k^2$ substitutions. So our goal is to design an algorithm that runs in truly sublinear time for $1 \ll k < \sqrt{n}$. 

There are two most relevant algorithms in the literature for this setting:
\begin{itemize}
\item \cite{AO12-edit} (building on~\cite{OR07-edit}) gave an algorithm that runs in time $n^{2+o(1)}/k^3$; in particular, it is sublinear for $k \gg n^{1/3}$.
\item \cite{GKS19} gave an algorithm that runs in time $\tO(n/k + k^3)$; in particular, it is truly sublinear for $k \ll n^{1/3}$.
\end{itemize}
In particular, \cite{GKS19} left as an open problem obtaining a sublinear algorithm for $k\approx n^{1/3}$. 

Our main result is a very simple algorithm that runs in time $\tO(n/\sqrt{k})$ and hence is simultaneously sublinear for all relevant values of $k$. 

\begin{theorem*}[Main result (informal); see Theorem~\ref{thm:no-preprocessing}]
  We can distinguish between $\ED(A, B) \le k$ and $\ED(A, B) = \Omega(k^2)$ in $\tilde{O}(n/\sqrt{k})$ time with high probability.
\end{theorem*}

Our algorithm is better than~\cite{AO12-edit,GKS19} for $n^{2/7} \ll k \ll n^{2/5}$ (and is also arguably simpler than both).

\paragraph{Independent work of Kociumaka and Saha}
The open problem of Goldenberg, Krautgamer, and Saha~\cite{GKS19} 
was also independently resolved by Kociumaka and Saha~\cite{KS20sublinear}.
They use essentially the same main algorithm (Algorithm~\label{alg:jaggedmatch} below), but use substantially different techniques to implement approximate queries to the subroutine we call $\method{MaxAlign}_k$. Their running time ($\tO(n/k+k^2)$) is faster than ours in the regime where our algorithm is faster than~\cite{AO12-edit}. 

\subsubsection*{Edit distance with preprocessing: results and technical insights}
Our starting point for this paper is the recent work of~\cite{GRS20} that designed algorithms for {\em edit distance with preprocessing}, namely the algorithm consists of two phases:
\begin{description}
\item[Preprocessing] where each string is preprocessed separately; and 
\item[Query] where the algorithm has access to both strings and outputs of the preprocess phase.
\end{description}
A simple and efficient preprocessing procedure proposed by~\cite{GRS20} is to compute a hash table for every contiguous substring. In the query phase, this enables an $O(\log(n))$-time implementation of a subroutine that given indices $i_A,i_B$ returns the longest common (contiguous) substring of $A,B$ starting at indices $i_A,i_B$ (respectively).
We use a simple modification of this subroutine, that we call $\method{MaxAlign}_k$: given only an index $i_B$ for string $B$, it returns the longest common (contiguous) substring of $A,B$ starting at indices $i_A,i_B$ (respectively) for any $i_A \in [i_B-k,i_B+k]$. (It is not hard to see that for $k$-close strings, we never need to consider other choices of $i_A$~\cite{Ukkonen85}.)

Given access to a $\method{MaxAlign}_k$ oracle, we obtain the following simple greedy algorithm for $k$-vs-$k^2$ edit distance: 
Starting from pointer $i_B = 1$, at each iteration it advances $i_B$ to the end of the next longest common subsequence returned by $\method{MaxAlign}_k$. 

\begin{algorithm*}[h]
\caption*{{\bf Algorithm~\ref{alg:jaggedmatch}} $\method{GreedyMatch}(A, B, k)$}
  \begin{algorithmic}
    \STATE $i_B \leftarrow 1$
    \STATE \textbf{for} $e$ \textbf{from} $1$ \textbf{to} $2k+1$
    \STATE \ \ \ $i_B \leftarrow i_B + \max(\MaxAlign_k(A, B, i_B), 1)$
    \STATE \ \ \ \textbf{if} $i_B > n$ \textbf{then} \textbf{return}\ \ SMALL
    \RETURN LARGE
  \end{algorithmic}
\end{algorithm*}

Each increase of the pointer $i_B$ costs at most $2k$ in edit distance (corresponding to the freedom to choose $i_A \in [i_B-k,i_B+k]$). Hence if $i_B$ reaches the end of $B$ in $O(k)$ steps, then $\ED(A,B) \le O(k^2)$ and we can accept; otherwise the edit distance is $>k$ and we can reject.
The above ideas suffice to solve $k$-vs-$k^2$ gap edit-distance in $\tO(k)$ query time after polynomial preprocessing%
\footnote{The prepocessing can be made near-linear, but in this setting our algorithm is still dominated by that of~\cite{CGK16}.}.

Without preprocessing, we can't afford to hash the entire input strings. Instead, we subsample $\approx 1/k$-fraction of the indices from each string and compute hashes for the sampled subsequences.  If the sampled indices perfectly align (with a suitable shift in $[\pm k]$), the hashes of identical contiguous substrings will be identical, whereas the hashes of substrings that are $>k$-far (even in Hamming distance) will be different (w.h.p.). This error is acceptable since we already incur a $\Theta(k)$-error for each call of $\method{MaxAlign}_k$. This algorithm would run in $\tO(n/k)$ time%
\footnote{There is also an additive $\tO(k)$ term like in the preprocessing case, but it is dominated by $\tO(n/k)$ for $k < \sqrt{n}$.}, but there is a caveat: when we construct the hash table, it is not yet possible to pick the indices so that they perfectly align (we don't know the suitable shift). 
Instead, we try $O(\sqrt{k})$ different shifts for each of $A,B$; by birthday paradox, there exists a pair of shifts that exactly adds up to the right shift in $[\pm k]$.
The total run time is given by $\tO(n/k \cdot \sqrt{k}) = \tO(n/\sqrt{k})$. 

\cite{GRS20} also considered the case where we can only preprocess one of the strings. In this case, we can mimic the strategy from the previous paragraph, but take all $O(k)$ shifts on the preprocessed string, saving the $O(\sqrt{k})$-factor at query time. This gives the following result:

  \begin{theorem*}[Informal statement of Theorem~\ref{thm:1-preprocessing}]
  We can distinguish between $\ED(A, B) \le k$ and $\ED(A, B) = \tilde{\Omega}(k^2)$ with high probability in $\tO(n)$ preprocessing time of $A$ and $\tilde{O}(n/k)$ query time.
\end{theorem*} 

Our query time improves over a $\tO(n/k+k^2)$-algorithm in~\cite{GRS20} that used similar ideas. (A similar algorithm with low asymmetric query complexity was also introduced in~\cite{GKS19}.)

\subsubsection*{Trading off running time for better approximation}

By combining our algorithm with the $h$-wave algorithm of~\cite{LMS98}, we can tradeoff approximation guarantee and running time in our algorithms. 
The running times we obtain for $k$ vs $k \ell$ edit distance are:
\begin{description}
\item[No preprocessing] $\tO(\frac{n\sqrt{k}+k^{2.5}}{\ell})$ running time for $\ell \in [\sqrt{k}, k]$. (Theorem~\ref{thm:no-preprocessing2})
\item[One-sided preprocessing] $\tO(\frac{nk}{\ell})$ preprocessing time and $\tO(\frac{n+k^2}{\ell})$ query time. (Theorem~\ref{thm:1-preprocessing2})
\item[Two-sided preprocessing] $\tO(\frac{nk}{\ell})$ preprocessing time and $\tO(\frac{k^2}{\ell})$ query time. (Corollary~\ref{thm:2-preprocessing2})
\end{description}

\subsection*{Organization}

Section~\ref{sec:prelim} gives an overview of the randomized hashing technique we use, as well as a structural lemma theorem for close strings. Section~\ref{sec:k-vs-k2} gives a meta-algorithm for distinguishing $k$ versus $k^2$ edit distance. Sections~\ref{subsec:two-sided},\ref{sec:zero-sided},\ref{sec:one-sided} respectively implement this meta-algorithm for two-, zero-, and one-sided preprocessing. %
Appendix~\ref{app:k-vs-k1+eps} explains how to trade off running time for improved gap of $k$ versus $k\ell$ edit distance. Appendix~\ref{app:omit} includes the proof of our structural decomposition lemma.

\section{Preliminaries}\label{sec:prelim}

\subsection{Rabin-Karp Hashing}

A standard preprocessing ingredient is Rabin-Karp-style rolling hashes (e.g., \cite{cormen2009introduction}). We identify the alphabet $\Sigma$ with $1, 2, \hdots, |\Sigma|$. Assume there is also $\$ \not\in \Sigma$, which we index by $|\Sigma|+1$.\footnote{We assume that all indices out of range of $A[1,n]$ are equal to $\$$.}  Assume before any preprocessing that we have picked a prime $p$ with $\Theta(\log n + \log |\Sigma|)$ digits as well a uniformly random value $x \in \{0, 1, \hdots, p-1\}$. We also have $S \subset [n]$, a \emph{subsample} of the indices which allows for sublinear preprocessing of the rolling hashes while still successfully testing string matching (up to a $\tilde{O}(n/|S|)$ Hamming error).

\begin{algorithm}[h]
\caption{$\method{InitRollingHash}(A, S)$}
  \begin{algorithmic}
    \INPUT $A \in \Sigma^n$; $S$ array of indices to be hashed
    \OUTPUT $H,$ a list of $|S|$+1 hashes
  \STATE $H \leftarrow [0]$
  \STATE $c \leftarrow 0$
  \STATE {\bf for} $i \in S$ {\bf then}\\
  \STATE \ \ \ $c \leftarrow cx + A[i]\mod p$
  \STATE \ \ \ append $c$ to $H$.
  \RETURN $H$
  \end{algorithmic}
\label{alg:inithash}
\end{algorithm}

\begin{algorithm}[h]
\caption{$\method{RetrieveRollingHash}(A, S, H, i, j)$}
  \begin{algorithmic}
    \INPUT $A \in \Sigma^n$; $S$ array of hashed indices; $H$ list of hashes; $i \le j$ indices from $1$ to $n$.
    \OUTPUT $h$, hash of string
  \STATE $i' \leftarrow$ least index such that $S[i'] \ge i$.
  \STATE $j' \leftarrow$ greatest index such that $S[j'] \le j$.
  \RETURN $h \leftarrow H[j'] - H[i'-1] x^{j'-i'+1} \mod p$
  \end{algorithmic}
\label{alg:retrievehash}
\end{algorithm}

Observe that $\method{InitRollingHash}$ runs in $\tilde{O}(|S|)$ time and $\method{RetrieveRollingHash}$ runs in $\tilde{O}(1)$ time. %
The correctness guarantees follow from the following standard proposition.

\begin{proposition}\label{prop:hash-standard}
  Let $A, B \in \Sigma^n$ and $S := \{1, 2, \hdots, n\}$. Let $H_A = \method{InitRollingHash}(A,S)$ and $H_B = \method{InitRollingHash}(B,S)$.  The following holds with probability at least $1-\frac{1}{n^4}$ over the choice of $x$. For all $i_A\le j_A$ and $i_B \le j_B$, we have that \[\method{RetrieveRollingHash}(A, S, H_A, i_A, j_A) = \method{RetrieveRollingHash}(A, S, H_B, i_B,j_B)\] if and only if $A[i_A, j_A] = B[i_B, j_B]$.
\end{proposition}

This claim is sufficient for our warm-up two-sided preprocessing algorithm. However, for the other algorithms, we need to have $|S| = o(n)$  for our hashing to be sublinear. This is captured by another claim.

\begin{claim}\label{claim:hash-random}
  Let $A, B \in \Sigma^n$ and $S \subseteq \{1, 2, \hdots, n\}$ be a random subset with each element included independently with probability at least $\alpha := \min(\tfrac{4\ln n}{k}, 1)$. Let $H_A = \method{InitRollingHash}(A,S)$ and $H_B = \method{InitRollingHash}(B,S)$. For any $i \le j$ in $\{1, \hdots, n\}$ we have 
  \begin{itemize}
  \item[(1)] If $A[i, j] = B[i, j]$ then $\method{RetrieveRollingHash}(A, S, H_A, i, j) = \method{RetrieveRollingHash}(B, S, H_B, i, j)$.
  \item[(2)] If $\Ham(A[i, j], B[i, j]) \ge k$ then with probability at least $1-\frac{1}{n^3}$ over the choice of $x$ and $S$, $\method{RetrieveRollingHash}(A, S, H_A, i, j) \neq \method{RetrieveRollingHash}(B, S, H_B, i, j)$ 
  \end{itemize}
\end{claim}

\begin{proof}
  Let $A_S$ and $B_S$ be the subsequences of $A$ and $B$ corresponding to the indices $S$. Note that if $A[i, j] = B[i, j]$ then $A_S[i', j'] = B_S[i', j']$, where $i'$ and $j'$ are chosen as in $\method{RetrieveRollingHash}$. Property (1) then follows by Proposition~\ref{prop:hash-standard}.

  If $\Ham(A[i, j], B[i, j]) \ge k$, the probability there exists $i_0 \in S \cap [i, j]$ such that $A[i_0] \neq B[i_0]$ and thus $A_S[i', j'] \neq B_S[i', j']$ is $1$ if $\alpha = 1$ and otherwise at least 
  \[1 - (1 - (4\ln n)/k)^{k} \ge 1 - 1/e^{4\ln n} = 1-1/n^4.\] 
  If $A_S[i', j'] \neq B_S[i', j']$ then by Proposition~\ref{prop:hash-standard}, \[\method{RetrieveRollingHash}(A, S, H_A, i, j) \neq \method{RetrieveRollingHash}(B, S, H_B, i, j)\] with probability at least $1 - 1/n^4$. Therefore, for a random $S$, $\method{RetrieveRollingHash}(A, S, H_A, i, j) \neq \method{RetrieveRollingHash}(B, S, H_B, i, j)$ is at least $1 - 1/n^4 - 1/n^4 > 1 - 1/n^3$. Thus, property (2) follows.
\end{proof}

\subsection{Structural Decomposition Lemma}

\begin{definition}[$k$-alignment and approximate $k$-alignment] \hfill

Given strings $A,B$, we say that a substring $B[i_B,i_B+d-1]$ with $1 \le i_B, i_B + d-1 \le n$ is in {\em $k$-alignment} in $A[i_A, i_A+d-1]$ if $|i_A-i_B|\le k$ and $A[i_A, i_A+d-1] = B[i_B,i_B+d-1]$. If instead we have $|i_A - i_B| \le 3k$ and $\ED(A[i_A, i_A+d-1], B[i_B, i_B+d-1]) \le 3k$, we say that $B[i_B, i_B+d-1]$ is in {\em approximate $k$-alignment} with $A[i_A, i_A+d-1]$.  We say that $B[i_B,i_B+d-1]$ has a (approximate) $k$-alignment in $A$ if there is an $i_A$ with $|i_A-i_B|\le k$ such that $B[i_B,i_B+d-1]$ is in (approximate) $k$-alignment with $A[i_A,i_A+d-1]$.
\end{definition}

For all our algorithms we need the following decomposition lemma. The proof is deferred to Appendix~\ref{app:omit}.

\begin{lemma}\label{lem:decomp}
  Let $A, B \in \Sigma^{*}$ be strings such that $\ED(A,B) \le k$. Then, $A$ and $B$ can be partitioned into at $2k+1$ intervals $I_1^A, \hdots, I_{2k+1}^A$; $I_1^B, \hdots, I_{2k+1}^B$, respectively, and a partial monotone matching $\pi : [2k+1] \to [2k+1] \cup \{\perp\}$ such that
  \begin{itemize}
  \item Unmatched intervals are of length at most $1$, and
  \item For all $i$ in the matching, $B[I_{\pi(i)}^B]$ is in $k$-alignment with $A[I_i^A]$.%
  \end{itemize}
\end{lemma}
 
\section{A meta-algorithm for distinguishing $k$ vs. $k^ 2$}\label{sec:k-vs-k2}

In this section, we present $\method{GreedyMatch}$ (Algorithm~\ref{alg:jaggedmatch}), a simple algorithm for distinguishing $\ED(A, B) \le O(k)$ from $\ED(A, B) \ge \Omega(k^2)$. The algorithm assumes access to data structure $\MaxAlign_k$ as defined below. In the following sections, we will present different implementations of this data structure for the case of two-sided, one-sided, and no preprocessing. 

Define $\MaxAlign_k(A, B, i_B)$ to be a function which returns $d \in [1, n]$. We say that an implementation of $\MaxAlign_k(A, B, i_B)$ is \emph{correct} if with probability $1$ it outputs the maximum $d$ such that $B[i_B, i_B+d-1]$ has a $k$-alignment in $A$, and if no $k$-alignment exists, it outputs $d = 0$. 
We say that an implementation is \emph{approximately correct} if the following are true.
\begin{enumerate}
\item Let $d$ be the maximal such that $B[i_B, i_B+d-1]$ has a $k$-alignment in $A$. With probability $1$, $\MaxAlign_k(A, B, i_B) \ge d$.
\item With probability at least $1-1/n^2$, $B[i_B, i_B+ \MaxAlign_k(A, B, i_B)-1]$ has an approximate $k$-alignment in $A$.
\end{enumerate}

We say that an implementation is \emph{half approximately correct} if the following are true.
\begin{enumerate}
\item Let $d$ be the maximal such that $B[i_B, i_B+d-1]$ has a $k$-alignment. With probability $1$, $\MaxAlign_k(A, B, i_B) > d/2$ (unless $d=0$).
\item With probability at least $1-1/n^2$, $B[i_B, i_B+ \MaxAlign_k(A, B, i_B)-1]$ has an approximate $k$-alignment in $A$.
\end{enumerate}

\begin{algorithm}[h]
\caption{$\method{GreedyMatch}(A, B, k)$}
  \begin{algorithmic}
    \INPUT $A, B \in \Sigma^n$, $k \le n$
    \OUTPUT SMALL if $\ED(A, B) \le k$ or LARGE if $\ED(A, B) > 40k^2$
    \STATE $i_B \leftarrow 1$
    \STATE \textbf{for} $e$ \textbf{from} $1$ \textbf{to} $2k+1$
    \STATE \ \ \ $i_B \leftarrow i_B + \max(\MaxAlign_k(A, B, i_B), 1)$
    \STATE \ \ \ \textbf{if} $i_B > n$
    \STATE \ \ \ \ \ \ \textbf{return}\ \ SMALL
    \RETURN LARGE
  \end{algorithmic}
\label{alg:jaggedmatch}
\end{algorithm}

We now give the following correctness guarantee.

\begin{lemma}\label{lem:jagged-match}
If $\MaxAlign_k$ is approximately correct and $\ED(A, B) \le k$, then with probability $1$, $\method{GreedyMatch}(A, B, k)$ returns SMALL. If $\MaxAlign_k$ is half approximately correct and $\ED(A, B) \le k/(2\log n)$, then with probability $1$, $\method{GreedyMatch}(A, B, k)$ returns SMALL. If $\MaxAlign_k$ is (half) approximately correct and $\ED(A, B) > 40k^2$, then with probability $1 - \frac{1}{n}$, $\method{GreedyMatch}(A, B, k)$ returns LARGE. Further, $\method{GreedyMatch}(A, B, k)$ makes $O(k)$ calls to $\MaxAlign_k$ and otherwise runs in $O(k\log n)$ time.
\end{lemma}

\begin{proof}
  If $\MaxAlign_k$ is approximately correct and if $\ED(A, B) \le k$ then by Lemma~\ref{lem:decomp}, $B$ can be decomposed into $2k+1$ intervals such that they are each of length at most $1$ or they exactly match the corresponding interval $A$, up to a shift of $k$. In the algorithm, if $i_B$ is in one of these intervals, then $\MaxAlign_k$ finds the rest of the interval (and perhaps more). Then, the algorithm will reach the end of $B$ in $2k+1$ steps and output SMALL.

  Let $k' = k/(2\log n)$.  If $\MaxAlign_k$ is half approximately correct and $\ED(A, B) \le k'$ then by Lemma~\ref{lem:decomp}, $B$ can be decomposed into $2k'+1$ intervals such that they are each of length at most $1$ or they exactly match the corresponding interval $A$, up to a shift of $k$. In the algorithm, if $i_B$ is in one of these intervals, then $\MaxAlign_k$ finds more than half of the interval. Thus, it takes at most $\log n$ steps for the algorithm to get past each of the $2k'+1$ intervals. Thus, the algorithm will reach the end of $B$ in $(2k'+1)(\log n) < 2k+1$ steps and output SMALL.
  
  For the other direction, it suffices to prove that if the algorithm outputs SMALL then $\ED(A, B) \le 40k^2$. If $\MaxAlign_k$ is (half) approximately correct, and the algorithm outputs SMALL, with probability at least $1-1/n$ over all calls to $\MaxAlign_k$, there exists a decomposition of $B$ into $2k+1$ intervals such that each is either of length $1$ or has an approximate $k$-alignment in $A$. Thus, there exists a sequence of edit operations from $B$ to $A$ by
  \begin{enumerate}
  \item deleting the at most $2k+1$ characters of $B$ which do not match,
  \item modifying at most $3k$ characters within each interval of $B$, and
  \item adding/deleting $6k$ characters between each consecutive pair of exactly-matching intervals (and before the first and after the last interval), since each match had a shift of up to $3k$.
  \end{enumerate}
  This is a total of $2k+1 + 3k(2k+1) +6k(2k+2) \le 40k^2$ operations. Thus, if $\ED(A, B) > 40k^2$, $\method{GreedyMatch}(A, B, k)$ return LARGE with probability at least $1 - \frac{1}{n}$. The runtime analysis follows by inspection.
\end{proof}

By Lemma~\ref{lem:jagged-match}, it suffices to implement $\MaxAlign_k$ efficiently and with $1/\poly(n)$ error probability in various models.%

\section{Warm-up: two-sided Preprocessing} \label{subsec:two-sided}

As warm-up, we give an implementation of $\MaxAlign_k$ that first preprocesses $A$ and $B$ (separately) for $\poly(n)$ time%
\footnote{It is not hard to improve the preprocessing time to $\tO(n)$. We omit the details since this algorithm would still not be optimal for the two-sided preprocessing setting.}, and then implement $\MaxAlign_k$ queries in $O(\log(n))$ time. %

Algorithm~\ref{alg:twosidedpreproc} takes as input a string $A$ and produces $(H_A, T_A)$, the rolling hashes of $A$ and a collection of hash tables. We let $H_B, T_B$ denote the corresponding preprocessing output for $B$. Algorithm~\ref{alg:align-twosidedpreproc} gives a correct implementation of $\MaxAlign_k$ with the assistance of this preprocessing.

\begin{algorithm}[h]
\caption{$\method{TwoSidedPreprocessing}_k(A)$}
  \begin{algorithmic}
    \INPUT $A \in \Sigma^n$, $k \le n$
    \OUTPUT $(H_A, T_A)$, a collection of hashes
    \STATE $H_A \leftarrow \method{InitRollingHash}(A, [1,n])$
    \STATE $\mathcal T_A \leftarrow n \times n$ matrix of hash tables
    \STATE \textbf{for} $i$ \textbf{from} $1$ \textbf{to} $n$
    \STATE \ \ \ \textbf{for} $j$ \textbf{from} $i$ \textbf{to} $n$
    \STATE \ \ \ \ \ \ \textbf{for} $a$ \textbf{from} $-k$ \textbf{to} $k$
    \STATE \ \ \ \ \ \ \ \ \ \textbf{if} $[i+a, j+a] \subset [1, n]$, add $\method{RetrieveRollingHash}(A, [1,n], H_A, i+a, j+a)$ to $T[i,j]$
    \RETURN $(H_A, T_A)$
  \end{algorithmic}
\label{alg:twosidedpreproc}
\end{algorithm}

\begin{algorithm}[h]
\caption{$\method{TwoSidedMaxAlign}_k(A, B, i_B)$}
  \begin{algorithmic}
    \INPUT $A \in \Sigma^n, B \in \Sigma^n$, $k \le n$, $i_B \in [1, n]$
    \OUTPUT $d \in [0, n]$.
    \STATE Binary search to find maximal $d \in [0, n-i_B+1]$ such that\\
    \ \ \ \ \ \ $\method{RetrieveRollingHash}(B, [1,n], H_B, i_B, i_B+d-1) \in T_A[i_B, i_B+d-1]$
    \RETURN $d$
  \end{algorithmic}
\label{alg:align-twosidedpreproc}
\end{algorithm}

\begin{lemma}\label{lem:two-sided-align}
 $\method{TwoSidedMaxAlign}_k$ is a correct implementation of $\MaxAlign_k$ . 
\end{lemma}
\begin{proof}
  Observe that $\method{TwoSidedMaxAlign}$ is correct if for all $a \in [-k, k]$, $\method{RetrieveRollingHash}(A, [1,n], H_A, i_B+a, i_B+d+a) = \method{RetrieveRollingHash}(B, [1,n], H_B, i_B, i_B+d)$ if and only if $A[i_B+a, i_B+d+a] = B[i_B,i_B+d]$.  By Claim~\ref{prop:hash-standard} and the union bound, this happens with probability at least $1 - \frac{1}{n^3}$.
  \end{proof}

\begin{theorem}\label{thm:zero-sided}
  When both $A$ and $B$ are preprocessed for $\poly(n)$ time, we can distinguish between $\ED(A, B) \le k$ and $\ED(A, B) > 40k^2$ in $\tilde{O}(k)$ time with probability $1 - \frac{1}{n}$.
\end{theorem}

\begin{remark}
Note that~\cite{CGK16}'s algorithm obtains similar guarantees while only spending $O(\log(n))$ query time.  Further, sketching algorithms for edit distance often achieve much better approximation factors, but the preprocessing is often not near-linear (e.g., \cite{BZ16}).\footnote{Document exchange (e.g., \cite{BZ16,H19}) is similar to the one-sided preprocessing model, but $A$ and $B$ are never brought together (rather a hash of $A$ is sent to $B$).}
\end{remark}

\begin{proof}[Proof of Theorem~\ref{thm:zero-sided}]
By Lemma~\ref{lem:two-sided-align}, $\method{TwoSidedMaxAlign}$ is correct (and thus approximately correct) so by Lemma~\ref{lem:jagged-match} succeeds with high enough probability that $\method{GreedyMatch}$ outputs the correct answer with probability at least $1 - \frac{1}{n}$.   
  
  By inspection, the preprocessing runs in $\poly(n)$ time. Further, as the binary search, hash computation, and table lookup are all $\tilde{O}(1)$ operations, $\method{TwoSidedMaxAlign}$ runs in $\tilde{O}(1)$ time, so the two-sided preprocessing version of $\method{GreedyMatch}$ runs in $\tilde{O}(k)$ time.
\end{proof}

\section{Main Result: $k$ vs $k^2$ with No Preprocessing} \label{sec:zero-sided}

As explained in the introduction, for the no preprocessing case, we take advantage of the fact that any $c \in [-k, k]$ can be written as $a\sqrt{k} + b$, there $a, b \in [-\sqrt{k}, \sqrt{k}]$.\footnote{We have $\sqrt{k}$ as shorthand for $\lceil \sqrt{k}\rceil$.} Thus, if for $A$ we compute a rolling hash tables according to $S + a\sqrt{k} := \{s + a\sqrt{k}, s \in S\} \cap [1, n]$ for $a \in \sqrt{k}$. Likewise, for $B$ we compute rolling hash tables according to $S - b := \{s - b, s \in S\} \cap [1, n]$. Then, if we seek to compare $A[i_B+c, i_B+c+d-1]$ and $B[i_B, i_B+d-1]$, it essentially suffices to compare\footnote{We need to ``shave'' $k$ from each end of the substrings as we need to ensure that $[i_B-b+k, i_B+d-1-b-k] \subset [i_B, i_B+d-1],$ etc.} $A[i_B+a\sqrt{k}+k, i_B+a\sqrt{k}+d-1-k]$ and $B[i_B-b+k, i_B+d-1-b-k]$. 

Before calling $\method{GreedyMatch}$, we call two methods $\method{ProcessA}$ and $\method{ProcessB}$ which compute these hash tables. Note that the procedures are asymmetrical. These take $\tilde{O}(n/\sqrt{k})$ time each.

\begin{algorithm}[h]
\caption{$\method{ProcessA}_k(A)$}
\begin{algorithmic}
  \INPUT $A \in \Sigma^n$
    \STATE \textbf{for} $a$ \textbf{from} $-\sqrt{k}$ \textbf{to} $\sqrt{k}$
    \STATE \ \ \ $H_{A,a\sqrt{k}} \leftarrow \method{InitRollingHash}(A, S+a\sqrt{k})$ 
    \RETURN $\{H_{A,a\sqrt{k}} : a \in [-\sqrt{k}, \sqrt{k}]\}$
  \end{algorithmic}
\label{alg:processA}
\end{algorithm}
\begin{algorithm}[h]
  
\caption{$\method{ProcessB}_k(B)$}
  \begin{algorithmic}
    \STATE \textbf{for} $b$ \textbf{from} $-\sqrt{k}$ \textbf{to} $\sqrt{k}$
    \STATE \ \ \ \ \ \ $H_{B,b} \leftarrow \method{InitRollingHash}(B, S-b)$
    \RETURN $\{H_{B,b} : b \in [-\sqrt{k}, \sqrt{k}]\}$
  \end{algorithmic}
\end{algorithm}

\begin{algorithm}[h!]
\caption{$\method{MaxAlign}_k(A, B, i_B)$}
  \begin{algorithmic}
    \INPUT $A \in \Sigma^n, B \in \Sigma^n$, $k \le n$, $i_B \in [1, n]$
    \STATE $d_{0} \leftarrow 2k$, $d_1 \leftarrow n-i_B+1$
    \STATE \textbf{while} $d_0 \neq d_1$ \textbf{do}
    \STATE \ \ \ $d_{\text{mid}} \leftarrow \lceil(d_0 + d_1) / 2\rceil$
    \STATE \ \ \ \textbf{if} $d \le 2k$ \textbf{then return} \True{}
    \STATE \ \ \ $L_A, L_B \leftarrow 0$
    \STATE \ \ \ \textbf{for} $a$ \textbf{from} $-\sqrt{k}$ \textbf{to} $\sqrt{k}$
    \STATE \ \ \ \ \ \ $h \leftarrow \method{RetrieveRollingHash}(A, S+a\sqrt{k}, H_{A,a\sqrt{k}}, i_B+k+a\sqrt{k}, i_B+d_{\text{mid}}-k-1+a\sqrt{k})$
    \STATE \ \ \ \ \ \ append $h$ to $L_A$
    \STATE \ \ \ \textbf{for} $b$ \textbf{from} $-\sqrt{k}$ \textbf{to} $\sqrt{k}$
    \STATE \ \ \ \ \ \ $h \leftarrow \method{RetrieveRollingHash}(B, S-b, H_{B,b}, i_B+k-b, i_B+d_{\text{mid}}-k-1-b)$
    \STATE \ \ \ \ \ \ append $h$ to $L_B$
    \STATE \ \ \ sort $L_A$ and $L_B$
    \STATE \ \ \ \textbf{if} $L_A \cap L_B \neq \emptyset$
    \STATE \ \ \ \ \ \ \textbf{then} $d_0 \leftarrow d_{\text{mid}}$
    \STATE \ \ \ \ \ \ \textbf{else} $d_1 \leftarrow d_{\text{mid}}-1$.
    \RETURN $d_0$
  \end{algorithmic}
\label{alg:align-onesidedpreproc}
\end{algorithm}

\begin{lemma}\label{lem:two-sided-align}
  $\MaxAlign_k$ is approximately correct.
\end{lemma}
\begin{proof}
  First, consider any $d \ge 1$ such that $B[i_B, i_B+d-1]$ has a $k$-alignment in $A$. We seek to show that $\MaxAlign_k(A, B, i_B) \ge d$ with probability $1$. Note that the output of $\MaxAlign_k$ is always at least $2k$, so we may assume that $d > 2k$. By definition of $k$-alignment, there exists $c \in [-k, k]$ such that $A[i_B + c, i_B+d-1+c] = B[i_B, i_B+d-1]$. Note that there exists $a, b \in [-\sqrt{k}, \sqrt{k}]$ such that $a\sqrt{k} + b = c$ and so
  \[
    A[i_B + k + a\sqrt{k}, i_B + d-k-1+a\sqrt{k}] = B[i_B + k-b, i_B + d-k-1-b].
  \]
  By applying Claim~\ref{claim:hash-random}, we have with probability $1$ that
  \begin{align*}
     \method{RetrieveRollingHash}&(A, S+a\sqrt{k}, H_{A,a\sqrt{k}}, i_B + k + a\sqrt{k}, d-k-1+a\sqrt{k})\\&= \method{RetrieveRollingHash}(B, S-b, H_{B,b}, i_B+k-b, i_B+d-k-1-b).
  \end{align*}
  Therefore, in the implementation of $\MaxAlign_k(A, B, i_B)$, if $d_{\text{mid}} = d$, then $L_A$ and $L_B$ will have nontrivial intersection, so the output of the binary search will be at least $d$, as desired. Thus, $\MaxAlign_k(A, B, i_B)$ will output at least the length of the maximal $k$-alignment.

  Second, we verify that $\MaxAlign_k$ outputs an approximate $k$-alignment. Let $d$ be the output of $\MaxAlign_k$, either $d = 2k$, in which case $B[i_B, i_B+d-1]$ trivially is in approximate $k$-alignment with $A[i_B, i_B+d-1]$ or $d > 2k$. Thus, for that $d$, the binary search found that $L_A \cap L_B \neq \emptyset$ and so there exists $a, b \in [-\sqrt{k}, \sqrt{k}]$ such that
  \begin{align*}
     \method{RetrieveRollingHash}&(A, S+a\sqrt{k}, H_{A,a\sqrt{k}}, i_B + k + a\sqrt{k}, d-k-1+a\sqrt{k})\\&= \method{RetrieveRollingHash}(B, S-b, H_{B,b}, i_B+k-b, i_B+d-k-1-b).
  \end{align*}
  Applying Claim~\ref{claim:hash-random} over all $\tilde{O}(\sqrt{k}^2) = \tilde{O}(k)$ comparisons of hashes made during the algorithm, with probability at least $1 - 1/n^3$, we must have that
  \[
    \ED(A[i_B + k + a\sqrt{k}, d-k-1+a\sqrt{k}], B[i_B+k-b, i_B+d-k-1-b]) \le k.
  \]
  Let $c := a\sqrt{k} + b$%
   then we have that
  \[
    \ED(A[i_B + k + c-b, i_B+d-k-1+c-b], B[i_B+k-b, i_B+d-k-1-b]) \le k
  \]
  so
  \[
    \ED(A[i_B+c, i_B+d-1+c], B[i_B, i_B+d-1]) \le 3k.
  \]
  Since $c = a\sqrt{k} + b \in [-3k, 3k]$, we have that $B[i_B, i_B+d-1]$ has an approximate $k$-alignment, as desired.
  \end{proof}

\begin{theorem}\label{thm:no-preprocessing}
  For $k \le O(\sqrt{n})$, with no preprocessing, we can distinguish between $\ED(A, B) \le k$ and $\ED(A, B) > 40k^2$ in $\tilde{O}(n/\sqrt{k})$ time with probability at least $1 - \frac{1}{n}$.
\end{theorem}
\begin{proof}
By Lemma~\ref{lem:two-sided-align}, $\method{MaxAlign}_k$ is approximately correct so by Lemma~\ref{lem:jagged-match} succeeds with high enough probability that $\method{GreedyMatch}$ outputs the correct answer with probability at least $1 - \frac{1}{n}$.

  By inspection, both $\method{ProcessA}_k$ and $\method{ProcessB}_k$ run in $\tilde{O}(n/\sqrt{k})$ time in expectation. Further, $\method{MaxAlign}_k$ runs in $\tilde{O}(\sqrt{k})$ time, so $\method{GreedyMatch}$ runs in $\tilde{O}(n/\sqrt{k}+k^{3/2}) = \tilde{O}(n/\sqrt{k})$ time.
\end{proof}

\section{One-sided Preprocessing}  \label{sec:one-sided}

For the one-sided preprocessing, we desire to get near-linear preprocessing time. To do that, $\MaxAlign_k$ shall be half approximately correct rather than approximately correct.

Recall as before we preselect $S \subset [1,n]$ with each element included i.i.d.~with probability $q := \min(\frac{4\ln n}{k}, 1)$. Also assume that every multiple of $k$ is in $S$ and that $n-1$ is in $S$.%
This only increases the size of $S$ by $n/k$, and does not hurt the success probability of Claim~\ref{claim:hash-random}. To achieve near-linear preprocessing, we only store $ \method{RetrieveRollingHash}(A, S+a, H_{A,a}, i+a, i+2^{i_0}-1+a)$, when $(S+a) \cap [i+a, i+2^{i_0}-1+a]$ changes. This happens when $i \in (S+1) \cup (S - 2^{i_0}+1)$.

\begin{algorithm}[h]
\caption{$\method{OneSidedPreprocessA}_k(A)$}
  \begin{algorithmic}
    \STATE \textbf{for} $a$ \textbf{from} $-k$ \textbf{to} $k$
    \STATE \ \ \ $H_{A,a} \leftarrow \method{InitRollingHash}(A, S+a)$
    \STATE $\mathcal T_A \leftarrow$  $\lfloor \log n\rfloor \times \frac{n}{k}$ matrix of empty hash tables
    \STATE \textbf{for} $i_0$ \textbf{in} $[\lfloor \log n \rfloor]$
    \STATE \ \ \ \textbf{for} $a$ \textbf{from} $-k$ \textbf{to} $k$
    \STATE \ \ \ \ \ \ \textbf{for} $i$ \textbf{in} $((S+1) \cup (S - 2^{i_0}+1))$ \textbf{ with } $[i+a, i+2^{i_0}-1+a] \subset [n]$
    \STATE \ \ \ \ \ \ \ \ \ $h \leftarrow \method{RetrieveRollingHash}(A, S+a, H_{A,a}, i+a, i+2^{i_0}-1+a)$
    \STATE \ \ \ \ \ \ \ \ \ add $h$ to $T_A[i_0, \lfloor i/k\rfloor - 1]$.
    \STATE \ \ \ \ \ \ \ \ \ add $h$ to $T_A[i_0, \lfloor i/k\rfloor]$.
    \STATE \ \ \ \ \ \ \ \ \ add $h$ to $T_A[i_0, \lfloor i/k\rfloor + 1]$.
    \RETURN $T_A$
  \end{algorithmic}
\label{alg:processA}
\end{algorithm}

\begin{claim}\label{claim:preproc-fast}
  $\method{OneSidedPreprocesA}(A)$ runs in $\tilde{O}(n)$ time in expectation.
\end{claim}

\begin{proof}
  Computing $\method{InitRollingHash}(A, S+a)$ 
  takes $|S| = \tilde{O}(n/k)$ time in expectation. Thus, computing the $H_{A,a}$'s 
  takes $\tilde{O}(n)$ time. The other loops take (amortized) $\tilde{O}(1) \cdot O(k) \cdot \tilde{O}(n/k) = \tilde{O}(n)$ time.
\end{proof}

Before we call $\method{GreedyMatch}$, we need to initialize the hash function for $B$ using $\method{OneSidedProcessB}(B)$. This takes $\tilde{O}(n/k)$ time in expectation.

\begin{algorithm}[h]
\caption{$\method{OneSidedProcessB}(B)$}
  \begin{algorithmic}
    \RETURN $H_{B} \leftarrow \method{InitRollingHash}(B, S)$
  \end{algorithmic}
\end{algorithm}

\begin{algorithm}[h]
\caption{$\method{OneSidedMaxAlign}_k(A, B, i_B)$}
  \begin{algorithmic}
    \INPUT $A \in \Sigma^n, B \in \Sigma^n$, $k\le n$,$i_b \in [1, n]$
    \STATE \textbf{for} $d \in [2^{\lfloor \log n\rfloor}, 2^{\lfloor \log n\rfloor-1}, \hdots, 1]$
    \STATE \ \ \ \textbf{if} $\method{RetrieveRollingHash}(B, S, H_B, i_B, i_B+d-1) \in T_A[\log d, \lfloor i_B/k\rfloor ]$ \textbf{then} \textbf{return} \ $d$
    \RETURN 0
  \end{algorithmic}
\label{alg:pseudoalign-onesidedpreproc}
\end{algorithm}

\begin{lemma}\label{lem:one-sided-pseudo-align}
 $\method{OneSidedMaxAlign}_{k}$ is half approximately correct. 
\end{lemma}

\begin{proof}
  First,  consider the maximal $d' \ge 1$ a power of two such that $B[i_B, i_B+d'-1]$ has a $k$-alignment in $A$. We seek to show that $\method{OneSidedMaxAlign}_k(A, B, i_B) \ge d'$ with probability $1$. By definition of $k$-alignment, there exists $a\in [-k, k]$ such that $A[i_B + a, i_B+d'-1+a] = B[i_B, i_B+d'-1]$. 
  By applying Claim~\ref{claim:hash-random}, we have with probability $1$ that
  \begin{align*}
     \method{RetrieveRollingHash}&(A, S+a, H_{A,a}, i_B + a, i_B+d'-1+a)\\&= \method{RetrieveRollingHash}(B, S, H_B, i_B, i_B+d'-1).
  \end{align*}
  Let $i'_B$ be the least integer in $((S+1) \cup (S-d'+1)) \cap [n]$ which is at least $i_B$. Since $S$ contains every multiple of $k$ (and $n-1$), we must have that $|i'_B - i_B| \le k$. Therefore,
  \begin{align*}
    \method{RetrieveRollingHash}&(A, S+a, H_{A,a}, i_B + a, i_B+d'-1+a)\\
                                &= \method{RetrieveRollingHash}(A, S+a, H_{A,a}, i'_B + a, i'_B+d'-1+a)\\
                                &\in T_A[\log d, \lfloor i'_B/k\rfloor + \{-1, 0, 1\} ].
  \end{align*} 

  Since $\lfloor i'_B/k\rfloor - \lfloor i_B/k\rfloor \in  \{-1, 0, 1\}$. We have that
  if $d = d'$, $\method{RetrieveRollingHash}(B, S, H_B, i_B, i_B+d'-1) \in T_A[\log d, \lfloor i_B/k\rfloor ]$. Thus, $\method{OneSidedMaxAlign}_k(A, B, i_B)$ 
 will output at least more than half the length of the maximal $k$-alignment.

  Second, we verify that $\method{OneSidedMaxAlign}_k$ outputs an approximate $k$-alignment. Let $d$ be the output of $\method{OneSidedMaxAlign}_k$, either $d = 0$, in which case $B[i_B, i_B+d-1]$ trivially is in approximate $k$-alignment with $A[i_B, i_B+d-1]$ or $d \ge 1$. Thus, for that $d$, the search found that $\method{RetrieveRollingHash}(B, S, H_B, i_B, i_B+d'-1) \in T_A[\log d, \lfloor i_B/k\rfloor ]$. Thus, there exists, $i'_B$ with $|\lfloor i'_B/k\rfloor - \lfloor i_B/k\rfloor| \le 1$ and $a \in [-k, k]$ such that
  \begin{align*}
    \method{RetrieveRollingHash}&(A, S+a, H_{A,a}, i'_B + a, i'_B+d'-1+a)\\&= \method{RetrieveRollingHash}(B, S, H_B, i_B, i_B+d'-1).
  \end{align*}
  Applying Claim~\ref{claim:hash-random} over all $\tilde{O}(k)$ potential comparisons of hashes made during the algorithm, with probability at least $1 - 1/n^3$, we must have that
  \[
    \ED(A[i'_B + a, i'_B+a+d'-1], B[i_B, i_B+d'-1]) \le k.
  \]
  Note that $|i'_B + a - i_B| \le |i'_B - i_B| + |a| \le 3k$. Thus $B[i_B, i_B+d'-1]$ has an approximate $k$-alignment, as desired.
  \end{proof}

  \begin{theorem}\label{thm:1-preprocessing}
  For all $A, B \in \Sigma^n$. When $A$ is preprocessed for $\tilde{O}(n)$ time in expectation, we can distinguish between $\ED(A, B) \le k/(2\log n)$ and $\ED(A, B) > 40k^2$ in $\tilde{O}(n/k)$ time with probability at least $1 - \frac{1}{n}$ over the random bits in the preprocessing (oblivious to $B$).
\end{theorem}
\begin{proof}
By Lemma~\ref{lem:one-sided-pseudo-align}, $\method{OneSidedMaxAlign}_k$ is half approximately correct so by Lemma~\ref{lem:jagged-match} succeeds with high enough probability that $\method{GreedyMatch}$ outputs the correct answer with probability at least $1 - \frac{1}{n}$.

  By Claim~\ref{claim:preproc-fast}, the preprocessing runs in $\tilde{O}(n)$ time. Also $\method{OneSidedProcessB}$ runs in $\tilde{O}(n/k)$ time. Further, $\method{OneSidedMaxAlign}_k$ runs in $\tilde{O}(1)$ time, as performing the power-of-two search, computing the hash, and doing the table lookups are $\tilde{O}(1)$ operations), so the one-sided preprocessing version of $\method{GreedyMatch}$ runs in $\tilde{O}(n/k+k) = \tilde{O}(n/k)$ time.
\end{proof}

\bibliographystyle{alpha}
\bibliography{sublinear}

\appendix

\section{Trading off running time for better approximation}\label{app:k-vs-k1+eps}

In this appendix, we show how to extend the results of the main body to distinguishing edit distance $k$ vs. $k\ell$. %

\subsection{Preliminaries: the $h$-wave algorithm}

The works of \cite{LMS98,GRS20} show that if one preprocesses both $A$ and $B$, then the \emph{exact} edit distance between $A$ and $B$ can be found by the following $O(k^2)$-sized dynamic program (called an \emph{h-wave}). The DP state is represented by a table $h[i,j]$, where $i \in [0, k]$ and $j \in [-k, k]$, initialized with $h[0,0] = 0$  and $h[0, j] = -\infty$ %
for all $j \in [-k,k]\setminus\{0\}$. The transitions for $i \ge 1$ are
\[
  h[i,j] = \max\begin{cases}
    h[i-1,j-1] + 1\\
    h[i-1,j] + \max(d, 1)\\
    h[i-1,j+1]
  \end{cases}  
\]
where $d$ is maximal such that $A[h[i-1,j]+1,h[i-1,j]+d] = B[h[i-1,j]+j+1,h[i-1,j]+j+d]$. Intuitively, $h[i,j]$ is the farthest length $n'$ such that $A[1,n']$ and $B[1,n'-j]$ have edit distance at most $i$. Then, $\ED(A, B) \le k$ if and only if $h[k,0] = n$.

\subsection{Approximate $h$-wave and $\MaxShiftAlign_{\ell}$}

We speed-up the original $h$-wave algorithm by considering a sparsified $h$-wave, where we store $h[i,j]$ with $i \in [0, k]$ and $j \in [-k, k]$ such that $j$ is a multiple of $\ell$. We again initialize $h[0, j] = 0$ for all $j$, but now we have the following transitions.

\[
  h[i,j] = \max\begin{cases}
    h[i-1,j-\ell] + \ell & \text{if\;} j-\ell \ge -k\\
    h[i-1,j] + \max(d, \ell)\\
    h[i-1,j+\ell] + \ell & \text{if\;} j+\ell \le k
  \end{cases}  
\]
where $d$ is maximal such that $A[h[i-1,j]+1+a,h[i-1,j]+j+a] = B[h[i-1,j]+j+1,h[i-1,j]+j+d]$ 
for some $a \in [-\ell, \ell]$. Note that when $\ell=1$ this mostly aligns with the $h$-wave algorithm except we have $h[i-1,j+\ell] + \ell$ instead of $h[i-1,j+\ell]$ (as we are only seeking an approximation, we do this to make the analysis simpler).

For the approximate $h$-wave, it is approximately true that $h[i,j]$ is the farthest length $n'$ such that $A[1,n']$ and $B[1,n'-j]$ have edit distance at most $\tilde{O}(i\ell)$ (see Lemma~\ref{lem:jaggedwave-pseudo1} and \ref{lem:jaggedwave-pseudo2}) for more details). Then, if $\ED(A, B) \le \tilde{O}(k)$, we have $h[k,0] \geq n$; and if $h[k, 0] < n$ then  $\ED(A, B) \ge \tilde{\Omega}(k\ell).$

Note that unlike the main body, we are no longer checking for matches where $A$ and $B$ have a common start point. Instead we require a generalization of $\MaxAlign_k$, which we call $\MaxShiftAlign_{\ell,k}(A, B, i_A, i_B)$. This algorithm finds the greatest positive integer $d$ such that $A[i_A+c, i_A+d-1 + c] = B[i_B, i_B+d-1]$ for some $c \in [-\ell, \ell]$, given the promise that $|i_A - i_B| \le k$.

\begin{definition}[shifted $(i_A, \ell)$-alignment and approximate shifted $(i_A,\ell)$-alignment] \hfill

Given strings $A,B$, and $i_A, i_B \in [1, n]$  we say that a $B[i_B,i_B+d-1]$ has a {\em shifted-$(i_A,\ell)$-alignment} with $A$ if there is $i$ with $|i_A-i|\le \ell$ and $A[i, i+d-1] = B[i_B,i_B+d-1]$. If instead we have that\footnote{Note that a shifted-$(i_A,\ell)$-alignment implies an approximate shifted-$(i_A,k)$-alignment, because $\ED(A[i_A, i_A+d-1], [i, i+d-1]) \le 2\ell$ if $|i-i_A| \le \ell$.} $\ED(A[i_A, i_A+d-1], B[i_B, i_B+d-1]) \le 10\ell$, we say that $B[i_B, i_B+d-1]$ has an {\em approximate shifted-$(i_A,\ell)$-alignment} with $A.$
\end{definition}

We say that an implementation of $\MaxShiftAlign_{\ell,k}(A, B, i_A, i_B)$ is \emph{approximately correct} if whenever $|i_A-i_B| \le k$ the following are true.
\begin{enumerate}
\item Let $d'$ be the maximal $d'$ such that $B[i_B, i_B+d'-1] = A[i_A, i_A+d'-1]$  for. With probability $1$, $\MaxShiftAlign_{\ell,k}(A, B, i_A, i_B) > d'/2$ (unless $d'=0$).
\item With probability at least $1-1/n^2$, $B[i_B, i_B+ \MaxShiftAlign_{\ell,k}(A, B, i_A, i_B)-1]$ has an approximate shifted-$(i_A, \ell)$-alignment in $A$.
\end{enumerate}

\begin{algorithm}[h]
\caption{$\method{GreedyWave}(A, B, k, \ell)$}
  \begin{algorithmic}
    \INPUT $A, B \in \Sigma^n$, $\ell \le k \le n$
    \OUTPUT SMALL if $\ED(A, B) \le \tilde{O}(k)$ or LARGE if $\ED(A, B) \ge \tilde{\Omega}(k \ell)$
    \STATE $h \leftarrow$ matrix with indices $[0, k] \times ([-k, k] \cap \ell \mathbb Z)$
    \STATE \textbf{for} $i$ \textbf{from} $0$ to $k$
    \STATE \ \ \ \textbf{for} $j$ multiples of $\ell$ \textbf{from} $-k$ to $k$
    \STATE \ \ \ \ \ \ \textbf{if} i = 0 \textbf{then} $h[i,j] \leftarrow -\infty$, $h[0,0]\leftarrow 0$.
    \STATE \ \ \ \ \ \ \textbf{else}
    \STATE \ \ \ \ \ \ \ \ \ $h[i,j] \leftarrow h[i-1,j] + \ell$
    \STATE \ \ \ \ \ \ \ \ \ \textbf{if}  $j-\ell \ge -k$ \textbf{then} $h[i,j] \leftarrow \max(h[i,j], h[i-1,j-\ell] + \ell)$
    \STATE \ \ \ \ \ \ \ \ \ \textbf{if}  $j+\ell \le k$ \textbf{then} $h[i,j] \leftarrow \max(h[i,j], h[i-1,j+\ell] + \ell)$
    \STATE \ \ \ \ \ \ \ \ \ $h[i,j] \leftarrow \max(h[i,j], h[i-1,j] + \MaxShiftAlign_{\ell,k}(A, B, h[i-1,j]+1, h[i-1,j]+j+1))$
    \STATE \textbf{if} $h[k,0] \ge n$ \textbf{return}\ \ SMALL
    \RETURN LARGE
  \end{algorithmic}
\label{alg:jaggedwave}
\end{algorithm}

\subsection{Analysis of $\method{GreedyWave}$}

We first prove that if we do not take any of the ``shortcuts'' given by $\MaxShiftAlign_{\ell,k}(A, B, h[i-1,j]+1, h[i-1,j]+j+1))$, we still increase $h$ by a quantifiable amount.

\begin{claim}\label{claim:jump}
  Consider $(i,j), (i',j') \in [0, k] \times ([-k, k] \cap \ell \mathbb Z)$ such that $i' \ge i$ and
  \[
    |j' - j| \le \ell(i' - i),
  \]
  then,
  \[
    h[i', j'] \ge h[i, j] + \ell(i' - i).
  \]
\end{claim}

\begin{proof}
  We prove this by induction on $i' - i$. The base case of $i'-i = 0$ is immediate.

  Now assume $i'-i \ge 1$ and that $|j' - j| \le \ell(i' - i)$. Note then there exists $e \in \{-1, 0, 1\}$ such that $j'+e\ell$ is between $j'$ and $j$ and
  \[
    |(j'+e\ell) - j| \le \ell(i'-1-i).
  \]
  By the induction hypothesis, we know then that
  \[
    h[i'-1, j'+e\ell] \ge h[i, j] +\ell(i'-i-1).
  \]
  Note then from $\method{GreedyWave}$, since $i' \ge 1$, we have that
  \[
    h[i', j'] \ge h[i'-1, j'+e\ell] + \ell.
  \]
  Therefore, combining the previous two inequalities.
  \[
    h[i',j'] \ge h[i, j] + \ell(i'-i).
  \]
\end{proof}

Correctness of Algorithm~\ref{alg:jaggedwave} is proved in the following pair of lemmas.

\begin{lemma}\label{lem:jaggedwave-pseudo1}

  Assume that $\MaxShiftAlign_{\ell,k}$ 
  is approximately correct. If $\ED(A, B) \le k/(20\lceil \log n\rceil )$, then $\method{GreedyWave}(A, B, k, \ell)$ outputs SMALL with probability $1$.
\end{lemma}

\begin{lemma}\label{lem:jaggedwave-pseudo2}

 Assume that $\MaxShiftAlign_{\ell,k}$ 
 is approximately correct. If $\ED(A, B) > 10 k\ell$, then $\method{GreedyWave}(A, B, k, \ell)$ 
 outputs LARGE with probability at least $1-1/n$.

\end{lemma}

\begin{proof}[Proof of Lemma~\ref{lem:jaggedwave-pseudo1}]
  \subsubsection*{Notation and Inductive Hypothesis}  Let $k'=\lfloor k/(20\lceil \log n \rceil)\rfloor$. Assume that $\ED(A, B) \le k'$. Let $I_1^A, \hdots, I_{2k'+1}^A$ and $I_1^B, \hdots, I_{2k'+1}^B$ and $\pi : [2k'+1] \to [2k'+1]\cup \{\perp\}$ be as in Lemma~\ref{lem:decomp}.
  Let $(a_1, b_1), \hdots, (a_t, b_t) \in \pi$ be the matching, ordered such that $a_1 \le \cdots \le a_t$ and $b_1 \le \cdots \le b_t$. Let $A_i = \max I^A_{a_i}$ and $B_i = \max I^B_{b_i}$.

  For the boundary, we let $a_0, b_0 = 0$, $A_0 = B_0 = 0$. Also let $a_{t+1}, b_{t+1} = 2k'+2$ and $A_{t+1} = B_{t+1} = n+1$. Let $r : \mathbb Z \to \ell\mathbb Z$ be the function which rounds each integer to the nearest multiple of $\ell$ (breaking ties by rounding down).
  
  It suffices to prove by induction for all $i \in \{0, 1, \hdots, t+1\}$, we have that \[
   h[a_i+b_i+(\lceil \log n\rceil + 1) i, r(A_i -B_i)] \ge A_i.
    \]

    The base case of $i = 0$ follows from $h[0,0] = 0$ in the initialization. Assume now that
    \[
 \textbf{Inductive hypothesis. } h[a_i+b_i+(\lceil \log n\rceil + 1)i, r(A_i -B_i)] \ge A_i.\]

    We seek to show that \[h[a_{i+1}+b_{i+1}+(\lceil \log n + 1\rceil)(i+1), r(A_{i+1} -B_{i+1})] \ge A_{i+1}.\]

  We complete the induction in two steps.
  
    \subsubsection*{Step 1, $h[a_{i+1}+b_{i+1}+(\lceil \log n\rceil +1)i+1, r(A_{i+1} -B_{i+1})] \ge A_i + a_{i+1}-a_i+1.$}

  First note that
  \begin{align*}
    |(A_{i+1}-B_{i+1}) -(A_i-B_i)| &= |(A_{i+1}-A_i-|I^A_{a_{i+1}}|)-(B_{i+1}-B_i-|I^B_{b_{i+1}}|)|\\
                                   &\le \left|A_{i+1}-A_i-|I^A_{a_{i+1}}|\right| + \left|B_{i+1}-B_i - |I^B_{b_{i+1}}| \right|\\
    &\le (a_{i+1}-a_i) + (b_{i+1}-b_i).
  \end{align*}

  Therefore,
  \begin{align*}
    |r(A_{i+1} -B_{i+1}) - r(A_i -B_i)| &\le (a_{i+1}-a_i) + (b_{i+1}-b_i) + \ell\\
                                        &\le \ell[(a_{i+1}-a_i) + (b_{i+1}-b_i) + 1]\\
                                        &= \ell[(a_{i+1}+b_{i+1}+(\lceil \log n\rceil +1)i+1) - (a_{i}+b_{i}+(\lceil \log n\rceil +1)i)].
  \end{align*}
  Therefore, we may apply Claim~\ref{claim:jump} to get that.
  \begin{align*}
    &\!\!\!\!\!\!h[a_{i+1}+b_{i+1}+(\lceil \log n\rceil +1)i+1, r(A_{i+1} -B_{i+1})]\\
    &\ge h[a_{i}+b_{i}+(\lceil \log n\rceil +1)i, r(A_{i} -B_{i})]\\
    &\ \ \ \ \ \ \ + \ell[(a_{i+1}+b_{i+1}+(\lceil \log n\rceil +1)i+1) - (a_{i}+b_{i}+(\lceil \log n\rceil +1)i)]\\
    &= A_i + \ell((a_{i+1}-a_i) + (b_{i+1}-b_i) + 1) \text{ (induction hypothesis)}\\
    &\ge A_i + a_{i+1}-a_i+1.
  \end{align*}

  \subsubsection*{Step 2, $h[a_{i+1}+b_{i+1}+(\lceil \log n\rceil +1)(i+1), r(A_{i+1} -B_{i+1})] \ge A_{i+1}. $}
  
  For $j \in [\lceil \log n + 1\rceil]$, let \[\hat{h}_{j} := h[a_{i+1}+b_{i+1}+(\lceil \log n\rceil)i+j, r(A_{i+1} -B_{i+1})].\]  If $\hat{h}_{j} \ge A_{i+1}$ for some $j \in [\lceil \log n \rceil]$, then we know that \[h[a_{i+1}+b_{i+1}+(\lceil \log n\rceil + 1)(i+1), r(A_{i+1} -B_{i+1})] \ge A_{i+1},\] which finishes the inductive step.
  
  Otherwise, we know that $\hat{h}_{j} \in [A_i + a_{i+1}-a_i+1, A_{i+1})$ for all $j \in [\lceil \log n \rceil]$. If $i = t$, then $A_t + a_{t+1} - a_t + 1 \ge n + 1 =  A_{t+1}$, because each interval strictly between $a_{t}$ to $a_{t+1}$ has length at most $1$. Therefore, $\hat{h}_{j} \ge A_{t+1}$, so we are done in this case.

  Now assume $i < t$ and consider any $j \in [\log n]$. Since every interval between $I^A_{a_i}$ and $I^A_{a_{i+1}}$ has length at most $1$, we have that $A_{i+1}-|I^A_{a_{i+1}}| + 1 \le A_i + a_{i+1}-a_i+1$. Therefore, $\hat{h}_{j} \in [A_{i+1} - |I^A_{a_{i+1}}| +1, A_i).$  Therefore, $A[\hat{h}_{j}, A_{i+1}] = B[h'-A_{i+1}+B_{i+1}, B_{i+1}]$. By definition of $r$, $|r(A_{i+1}-B_{i+1}) - A_{i+1}+B_{i+1}| \le \ell$. Since $\MaxShiftAlign_{\ell,k}$ is approximately correct, \[\MaxShiftAlign_{\ell,k}(A, B, \hat{h}_{j}, \hat{h}_{j}+r(A_{i+1}-B_{i+1})) > (A_{i+1} - \hat{h}_{j})/2.\] Therefore, \[\hat{h}_{j+1} \ge \hat{h}_{j} + d > \frac{A_{i+1}+\hat{h}_{j}}{2}.\] By composing these inequalities, we have that \[\hat{h}_{j+1} > \frac{(2^{j}-1)A_{i+1} + \hat{h}_{1}}{2^{j}}.\] Plugging in $j = \lceil \log n \rceil$, we get that \[\hat{h}_{\lceil \log n\rceil+1} > \frac{(n-1)}{n}A_{i+1} \ge A_{i+1}-1,\] so \[\hat{h}_{\lceil \log n \rceil + 1} = h[a_{i+1}+b_{i+1}+(\lceil \log n\rceil +1)(i+1), r(A_{i+1} -B_{i+1})] \ge A_{i+1},\] as desired.

  \subsubsection*{Conclusion.}
  Therefore, we have that $h[a_{t+2}+b_{t+2}+(\lceil \log n\rceil +1)(t+1), 0] \ge n+1$. Thus, we report SMALL as long as $a_{t+2} +b_{t+2} +(\lceil \log n\rceil +1)(t+1) \le k$. Observe that
  \begin{align*}
    a_{t+2}+b_{t+2}+(\lceil \log n\rceil +1)(t+1) &\le 2k'+2+2k'+2+(\lceil \log n \rceil + 1)(2k'+2)\\
                                                  &\le 20\lceil \log n\rceil k' \le k.
  \end{align*} Therefore, our algorithm always reports SMALL when $\ED(A, B) \le k/(20\lceil \log n\rceil).$ 
\end{proof}

\begin{proof}[Proof of Lemma~\ref{lem:jaggedwave-pseudo2}]
    Assume that $\method{GreedyWave}(A, B, k, \ell)$ output SMALL, and that $\MaxShiftAlign_{\ell,k}$ never failed at being approximately correct. For succinctness, we let $h_{i,j}$ be shorthand for $h[i,j]$

  We prove by induction that for all $i$ and $j$ for which $h_{i,j} \neq -\infty$, $\ED(A[1,h_{i,j}], B[1,h_{i,j}+j\ell]) \le 10i\ell + |j|$. This holds for the base case $i=j=0$. Now, we break into cases depending on how $h_{i,j}$ was computed.

  \subsubsection*{Case 1, $h_{i,j} = h_{i-1,j}+\ell$.}
  
  If $h_{i,j} = h_{i-1,j}+\ell$, then observe that
  \begin{align*}
    \ED(A[1,h_{i,j}], B[1,h_{i,j}+j]) &\le 2\ell + \ED(A[1,h_{i-1,j}], B[1,h_{i-1,j} +j])\\
                                    &\le 2\ell + 10(i-1)\ell+|j|\\
                                    &< 10i\ell + |j|.
  \end{align*}
  \subsubsection*{Case 2, $h_{i,j} = h_{i-1,j-\ell}+\ell$.}
  If $h_{i,j} = h_{i-1,j-\ell}+\ell$, then observe that
  \begin{align*}
    \ED(A[1,h_{i,j}], B[1,h_{i,j}+j]) &\le \ED(A[1,h_{i-1,j-\ell}+\ell], B[1,h_{i-1,j-\ell} +j+\ell])\\
                                          &\le 3\ell + \ED(A[1,h_{i-1,j-\ell}], B[1,h_{i-1,j-\ell} +j-\ell])\\
    &\le 3\ell + 10(i-1)\ell+|j-\ell| < 10i\ell+|j|.
  \end{align*}
    \subsubsection*{Case 3, $h_{i,j} = h_{i-1,j+\ell}+\ell$.}
   If $h_{i,j} = h_{i-1,j+\ell}+\ell$, then observe that
  \begin{align*}
    \ED(A[1,h_{i,j}], B[1,h_{i,j}+j\ell]) &\le \ED(A[1,h_{i-1,j+\ell}+\ell], B[1,h_{i-1,j+\ell} +j+\ell])\\
                                          &\le \ell + \ED(A[1,h_{i-1,j+\ell}], B[1,h_{i-1,j+\ell} +j+\ell])\\
    &\le \ell + 10(i-1)\ell+|j+\ell| < 10i\ell+|j|).
  \end{align*}
  \subsubsection*{Case 4, $h_{i,j} = h_{i-1,j} + d$}
  Finally, if $h_{i,j} = h_{i-1,j} + d$, then \[\ED(A[h_{i-1,j}+1, h_{i-1,j}+d], B[h_{i-1,j}+j+1, h_{i-1,j}+j+d]) \le 10\ell\] since $\MaxShiftAlign_{\ell,k}$ is approximately correct. Thus,
  \begin{align*}
    \ED(A[1,h_{i,j}], B[1,h_{i,j}+j]) &\le \ED(A[1,h_{i-1,j}], B[1,h_{i-1,j+\ell} +j])\\&+ \ED(A[h_{i-1,j}+1, h_{i-1,j}+d],\\&\ \ \ \ \  B[h_{i-1,j}+j+1, h_{i-1,j}+j+d])\\
                                          &\le 10(i-1)\ell+|j| + 10\ell < 10i\ell+|j|.
  \end{align*}

  This completes the induction. Therefore, since  $\method{GreedyWave}(A, B, k, \ell)$ output SMALL, we have that $h[k,0] \ge n$. Thus, $\ED(A, B) \le 10k\ell$, as desired.

\end{proof}

\subsection{Implementing $\MaxShiftAlign_{\ell,k}$}

\subsubsection{$\MaxShiftAlign_{\ell,k}$ with No Preprocessing}

We use a nearly-identical algorithm to that of Section~\ref{sec:zero-sided}, including looking at $\sqrt{k}$ shifts for both $A$ and $B$. But, we now sample $S \subset [1,n]$ so that each element is included with probability at least $\min(4\ln n/\ell, 1)$ (instead of $\min(4\ln n/k, 1)$).

\begin{algorithm}[h]
\caption{$\method{ProcessA}_{\ell,k}(A)$}
\begin{algorithmic}
  \INPUT $A \in \Sigma^n$
    \STATE \textbf{for} $a$ \textbf{from} $-2\sqrt{k}$ \textbf{to} $2\sqrt{k}$
    \STATE \ \ \ $H_{A,a\sqrt{k}} \leftarrow \method{InitRollingHash}(A, S+a\sqrt{k})$ 
    \RETURN $\{H_{A,a\sqrt{k}} : a \in [-2\sqrt{k}, 2\sqrt{k}]\}$
  \end{algorithmic}
\end{algorithm}
\begin{algorithm}[h]
  
\caption{$\method{ProcessB}_{\ell,k}(B)$}
  \begin{algorithmic}
    \STATE \textbf{for} $b$ \textbf{from} $-\sqrt{k}$ \textbf{to} $\sqrt{k}$
    \STATE \ \ \ \ \ \ $H_{B,b} \leftarrow \method{InitRollingHash}(B, S-b)$
    \RETURN $\{H_{B,b} : b \in [-\sqrt{k}, \sqrt{k}]\}$
  \end{algorithmic}
\end{algorithm}

\begin{algorithm}[h!]
\caption{$\method{MaxShiftAlign}_{\ell,k}(A, B, i_A, i_B)$}
  \begin{algorithmic}
    \INPUT $A \in \Sigma^n, B \in \Sigma^n$, $\ell, k \le n$, $\ell \ge \sqrt{k}$, $i_A, i_B \in [1, n], |i_A-i_B|\le k$
    \STATE $d_{0} \leftarrow 2\ell$, $d_1 \leftarrow n-i_b+1$
    \STATE \textbf{while} $d_0 \neq d_1$ \textbf{do}
    \STATE \ \ \ $d_{\text{mid}} \leftarrow \lceil(d_0 + d_1) / 2\rceil$
    \STATE \ \ \ \textbf{if} $d \le 2k$ \textbf{then return} \True{}
    \STATE \ \ \ $L_A, L_B \leftarrow 0$
    \STATE \ \ \ \textbf{for} $a$ \textbf{from} $\lfloor\frac{i_A-i_B-\ell}{\sqrt{k}}\rfloor$ \textbf{to} $\lceil \frac{i_A-i_B+\ell}{\sqrt{k}}\rceil $
    \STATE \ \ \ \ \ \ $h \leftarrow \method{RetrieveRollingHash}(A, S+a\sqrt{k}, H_{A,a\sqrt{k}}, i_B+\ell+a\sqrt{k}, i_B+d_{\text{mid}}-\ell-1+a\sqrt{k})$
    \STATE \ \ \ \ \ \ append $h$ to $L_A$
    \STATE \ \ \ \textbf{for} $b$ \textbf{from} $-\sqrt{k}$ \textbf{to} $\sqrt{k}$
    \STATE \ \ \ \ \ \ $h \leftarrow \method{RetrieveRollingHash}(B, S-b, H_{B,b}, i_B+\ell-b, i_B+d_{\text{mid}}-\ell-1-b)$
    \STATE \ \ \ \ \ \ append $h$ to $L_B$
    \STATE \ \ \ sort $L_A$ and $L_B$
    \STATE \ \ \ \textbf{if} $L_A \cap L_B \neq \emptyset$
    \STATE \ \ \ \ \ \ \textbf{then} $d_0 \leftarrow d_{\text{mid}}$
    \STATE \ \ \ \ \ \ \textbf{else} $d_1 \leftarrow d_{\text{mid}}-1$.
    \RETURN $d_0$
  \end{algorithmic}
\end{algorithm}

\begin{lemma}\label{lem:two-sided-shiftalign}
  $\MaxShiftAlign_{\ell,k}$ is approximately correct.
\end{lemma}
\begin{proof}
  First, consider any $d \ge 1$ such that $B[i_B, i_B+d-1]$ has a $(i_A,\ell)$-alignment with $A$. We seek to show that $\MaxShiftAlign_{\ell,k}(A, B, i_A,i_B) \ge d$ with probability $1$. Note that the output of $\MaxShiftAlign_{\ell,k}$ is always at least $2\ell$, so we may assume that $d > 2k$. By definition of $(i_A,\ell)$-alignment, there exists $c \in [-\ell, \ell]$ such that $A[i_A + c, i_A+c+d-1] = B[i_B, i_B+d-1]$. Note that since $|c+i_A-i_B| \le 2k$, there exists $a \in [-2\sqrt{k}, 2\sqrt{k}]$  and $b \in \sqrt{k}$ such that $a\sqrt{k} + b = c+i_A-i_B$. In fact, we may take $a \in \left[\lfloor\frac{i_A-i_B-\ell}{\sqrt{k}}\rfloor, \lceil \frac{i_A-i_B+\ell}{\sqrt{k}}\rceil\right].$ Thus, since $b \in [-\sqrt{k}, \sqrt{k}] \subset [-\ell, \ell]$, we have that
  \begin{align*}
    B[i_B + \ell-b, i_B + d-\ell-1-b] &=  A[i_A+c + \ell-b, i_A+c + d-\ell-1-b]\\
                                       &= A[i_B + \ell + a\sqrt{k}, i_B + d-\ell-1+a\sqrt{k}].
  \end{align*}
  By applying Claim~\ref{claim:hash-random}, we have with probability $1$ that
  \begin{align*}
     \method{RetrieveRollingHash}&(A, S+a\sqrt{k}, H_{A,a\sqrt{k}}, i_B + \ell + a\sqrt{k}, d-\ell-1+a\sqrt{k})\\&= \method{RetrieveRollingHash}(B, S-b, H_{B,b}, i_B+\ell-b, i_B+d-\ell-1-b).
  \end{align*}
  Therefore, in the implementation of $\MaxShiftAlign_{\ell,k}(A, B, i_B)$, if $d_{\text{mid}} = d$, then $L_A$ and $L_B$ will have nontrivial intersection, so the output of the binary search will be at least $d$, as desired. Thus, $\MaxShiftAlign_{\ell,k}(A, B, i_B)$ will output at least the length of the maximal $(i_A,\ell)$-alignment.

  Second, we verify that $\MaxShiftAlign_{\ell,k}$ outputs an approximate $(i_A,\ell)$-alignment. Let $d$ be the output of $\MaxShiftAlign_{\ell,k}$, either $d = 2\ell$, in which case $B[i_B, i_B+d-1]$ trivially is in approximate $(i_A,\ell)$-alignment with $A$ or $d > 2\ell$. Thus, for that $d$, the binary search found that $L_A \cap L_B \neq \emptyset$ and so there exists $a \in \left[\lfloor\frac{i_A-i_B-\ell}{\sqrt{k}}\rfloor, \lceil \frac{i_A-i_B+\ell}{\sqrt{k}}\rceil\right], b \in [-\sqrt{k}, \sqrt{k}]$ such that
  \begin{align*}
     \method{RetrieveRollingHash}&(A, S+a\sqrt{k}, H_{A,a\sqrt{k}}, i_B + \ell + a\sqrt{k}, d-\ell-1+a\sqrt{k})\\&= \method{RetrieveRollingHash}(B, S-b, H_{B,b}, i_B+\ell-b, i_B+d-\ell-1-b).
  \end{align*}
  Applying Claim~\ref{claim:hash-random} over all at most $\tilde{O}(\sqrt{k}^2) = \tilde{O}(k)$ comparisons of hashes made during the algorithm, with probability at least $1 - 1/n^3$, we must have that
  \[
    \ED(A[i_B + \ell + a\sqrt{k}, d-\ell-1+a\sqrt{k}], B[i_B+\ell-b, i_B+d-\ell-1-b]) \le \ell.
  \]
  Let $c := a\sqrt{k} + b$, then we have that
  \[
    \ED(A[i_B+c + \ell-b, i_B+c+d-\ell-1-b], B[i_B+\ell-b, i_B+d-\ell-1-b]) \le \ell
  \]
  Therefore,
  \[
    \ED(A[i_B+c, i_B+c+d-1], B[i_B, i_B+d-1]) \le 3\ell,
  \]
  Note that
  \begin{align*}
    i_B + c &= i_B + a\sqrt{k} +b \\
            &\ge i_B + [(i_A-i_B-\ell) - \sqrt{k}] - \sqrt{k}\\
            &\ge i_A - 3\ell.
  \end{align*}
  Likewise, $i_B +c \le i_A + 3\ell$. Therefore,
  \[
    \ED(A[i_B+c, i_B+c+d-1], A[i_A, i_A+d-1]) \le 6\ell,
  \]
  since we need at most $3\ell$ insertions and $3\ell$ deletions to go between the two strings. By the triangle inequality we then have that
  \[
    \ED(A[i_A, i_A+d-1], B[i_B, i_B+d-1]) < 10\ell,
  \]
  as desired.
\end{proof}

\begin{theorem}\label{thm:no-preprocessing2}
  If $\ell \ge \sqrt{k}$, with no preprocessing, we can distinguish between $\ED(A, B) \le \tilde{O}(k)$ and $\ED(A, B) \ge \tilde{\Omega}(\ell k)$ in $\tilde{O}(\tfrac{(n+k^2)\sqrt{k}}{\ell})$ time with probability at least $1 - \frac{1}{n}$.
\end{theorem}
\begin{proof}
By Lemma~\ref{lem:two-sided-shiftalign}, $\method{MaxShiftAlign}_{\ell,k}$ is approximately correct so by Lemmas~\ref{lem:jaggedwave-pseudo1} and~\ref{lem:jaggedwave-pseudo2} succeeds with high enough probability that $\method{GreedyMatch}$ outputs the correct answer with probability at least $1 - \frac{1}{n}$.

  Both $\method{ProcessA}_{\ell,k}$ and $\method{ProcessB}_{\ell,k}$ run in $\tilde{O}(\frac{n\sqrt{k}}{\ell})$ time. Further, $\method{MaxShiftAlign}_{\ell,k}$ runs in $\tilde{O}(\sqrt{k})$ time. Since $\method{GreedyWave}$ makes runs in $\tilde{O}(k^2/\ell)$ and makes that many calls to $\MaxShiftAlign_{\ell,k}$, we have that $\method{GreedyWave}$ runs in $\tilde{O}(\tfrac{(n+k^2)\sqrt{k}}{\ell})$ time.
\end{proof}

Setting $\ell = \tilde{\Theta}(k^{1-\eps})$ with $\eps < 1/2$, we get that distinguishing $k$ from $k^{2-\eps}$ can be done in $\tilde{O}(n k^{-1/2+\eps} + k^{3/2+\eps})$ time.

\subsubsection{$\MaxShiftAlign_{\ell,k}$ with One-Sided and Two-Sided Preprocessing}

For the one-sided preprocessing, we mimic \ref{sec:one-sided}. We preselect $S \subset [1,n]$ which each element included i.i.d.~with probability $q := \min(\frac{4\ln n}{\ell}, 1)$. Also assume that every multiple of $\ell$ (and $n-1$) is in $S$. This only increases the size of $S$ by $n/\ell$, and does not hurt the success probability of Claim~\ref{claim:hash-random}.

\begin{algorithm}[h]
\caption{$\method{OneSidedPreprocessA}_{\ell,k}(A)$}
  \begin{algorithmic}
    \STATE \textbf{for} $a$ \textbf{from} $-k$ \textbf{to} $k$
    \STATE \ \ \ $H_{A,a} \leftarrow \method{InitRollingHash}(A, S+a)$
    \STATE $\mathcal T_A \leftarrow$ list of $[\lfloor \log n\rfloor] \times [\tfrac{n}{\ell}] \times [-\frac{k}{\ell}-1, \frac{k}{\ell}+1]$ matrix of empty hash tables
    \STATE \textbf{for} $e_0$ \textbf{in} $[\lfloor \log n \rfloor]$
    \STATE \ \ \ \textbf{for} $a$ \textbf{from} $-k$ \textbf{to} $k$
    \STATE \ \ \ \ \ \ \textbf{for} $i$ \textbf{in} $((S+1) \cup (S - 2^{i_0}+1)) \cap [n]$
    \STATE \ \ \ \ \ \ \ \ \ $h \leftarrow \method{RetrieveRollingHash}(A, S+a, H_{A,a}, i+a, i+2^{e_0}-1+a)$
    \STATE \ \ \ \ \ \ \ \ \ add $h$ to $T_A[i_0, \lfloor i/\ell\rfloor+e_1, \lfloor a/\ell\rfloor+e_2]$ \textbf{for} $e_1,e_2 \in \{-1, 0, 1\}$.
    \RETURN $T_A$
  \end{algorithmic}
\label{alg:processA}
\end{algorithm}

\begin{claim}\label{claim:preproc-fast2}
  $\method{OneSidedPreprocessA}_{\ell,k}(A)$ runs in $\tilde{O}(\frac{nk}{\ell})$ time in expectation.
\end{claim}

\begin{proof}
  Computing $\method{InitRollingHash}(A-a, S)$ takes $|S| = \tilde{O}(n/\ell)$ time in expectation. Thus, computing the $H_{A-a}$'s takes $\tilde{O}(\frac{nk}{\ell})$ time. Initializing the hash table takes $\tilde{O}(\frac{nk}{\ell^2})$ time. The other loops take (amortized) $\tilde{O}(1) \cdot O(k) \cdot \tilde{O}(n/k) = \tilde{O}(n)$ time.
\end{proof}

Before we call $\method{GreedyWave}$, we need to initialize the hash function for $B$ using $\method{OneSidedProcessB}_{\ell,k}(B)$. This takes $\tilde{O}(n/\ell)$ time in expectation.

\begin{algorithm}[h]
\caption{$\method{OneSidedProcessB}_{\ell,k}(B)$}
  \begin{algorithmic}
    \RETURN $H_{B} \leftarrow \method{InitRollingHash}(B, S)$
  \end{algorithmic}
\end{algorithm}

\begin{algorithm}[h]
\caption{$\method{OneSidedMaxShiftAlign}_{\ell,k}(A, B, i_A,i_B)$}
  \begin{algorithmic}
    \INPUT $A \in \Sigma^n, B \in \Sigma^n$, $\ell,k\le n$, $i_A,i_B \in [1, n]$ $|i_A-i_B| \le k$.
    \STATE \textbf{for} $d_0 \in [2^{\lfloor \log n\rfloor}, 2^{\lfloor \log n\rfloor-1}, \hdots, 1]$
    \STATE \ \ \ \textbf{if} $\method{RetrieveRollingHash}(B, S, H_B, i_B, i_B+d_0-1) \in T_A[\log d_0, \lfloor i_B/\ell\rfloor , \lfloor (i_A-i_B)/\ell\rfloor]$\\
    \STATE \ \ \ \ \ \ \textbf{then} \textbf{return} \ $d_0$
    \RETURN 0
  \end{algorithmic}
\label{alg:pseudoalign-onesidedpreproc}
\end{algorithm}

\begin{lemma}\label{lem:one-sided-shiftalign}
 $\method{OneSidedMaxShiftAlign}_{\ell,k}$ is approximately correct. 
\end{lemma}

\begin{proof}
  First,  consider the maximal $d \ge 1$ a power of two such that $B[i_B, i_B+d-1]$ has a $(i_A,\ell)$-alignment with $A$. We seek to show that $\method{OneSidedMaxShiftAlign}_k(A, B,i_A, i_B) \ge d$ with probability $1$. By definition of $(i_A,\ell)$-alignment, there exists $c\in [-\ell, \ell]$ such that $A[i_A+c, i_A+c+d'-1] = B[i_B, i_B+d'-1]$. Let $a := i_A+c-i_B$.
  By applying Claim~\ref{claim:hash-random}, we have with probability $1$ that
  \begin{align*}
     \method{RetrieveRollingHash}&(A, S+a, H_{A,a}, i_B + a, i_B+d-1+a)\\&= \method{RetrieveRollingHash}(B, S, H_B, i_B, i_B+d-1).
  \end{align*}
  Let $i_B'$ be the least element of $(S+1) \cup (S-d+1)$ which is at least $i_B$. Since every multiple of $\ell$ is in $S$ (as well as $n-1$), we have that $|i'_B - i_B| \le \ell$ and
  \begin{align*}
    \method{RetrieveRollingHash}&(A, S+a, H_{A,a}, i_B + a, i_B+d-1+a)\\
                                &=  \method{RetrieveRollingHash}(A, S+a, H_{A,a}, i'_B + a, i'_B+d-1+a)\\
    &\in  T_A[\log d, \lfloor i'_B/\ell\rfloor , \lfloor a/\ell\rfloor].
  \end{align*}
  Therefore, in the implementation of $\method{OneSidedMaxShiftAlign}_{\ell,k}(A, B, i_A, i_B)$, if $d_0 = d$, we have that \[
    \method{RetrieveRollingHash}(B, S, H_B, i_B, i_B+d-1) \in T_A[\log d, \lfloor i'_B/\ell\rfloor +e_1, \lfloor a/\ell\rfloor + e_2].
  \]
  for all $e_1,e_2 \in \{-1, 0, 1\}$. Since $|i'_B - i_B| \le \ell$ and $|a - (i_A-i_B)| \le \ell$, we have that
  \[
    \method{RetrieveRollingHash}(B, S, H_B, i_B, i_B+d-1) \in T_A[\log d, \lfloor i_B/\ell\rfloor , \lfloor (i_A-i_B)/\ell\rfloor].
  \]
  Thus, $\method{OneSidedMaxShiftAlign}_{\ell,k}(A, B, i_A, i_B)$ 
  will output at least more than half the length of the maximal shift $(i_A,\ell)$-alignment.

  Second, we verify that $\method{OneSidedMaxShiftAlign}_{\ell,k}$ outputs an approximate $(i_A,\ell)$-alignment. Let $d$ be the output of $\method{OneSidedMaxShiftAlign}_{\ell,k}$. Either $d = 0$, in which case $B[i_B, i_B+d-1]$ trivially is in approximate $(i_A,\ell)$-alignment with $A$ or $d \ge 1$. Thus, for that $d$, the search found that $\method{RetrieveRollingHash}(B, S, H_B, i_B, i_B+d-1) \in T_A[\log d, \lfloor i_B/\ell\rfloor, \lfloor (i_A-i_B)/\ell\rfloor]$. Thus, there exists $i_A'$ with $|\lfloor i_A'/\ell\rfloor - \lfloor i_B/\ell\rfloor| \le 1$, and  $a \in [-k, k]$ such that \begin{align}|\lfloor a/\ell\rfloor - \lfloor (i_A-i_B)/\ell\rfloor| \le 1\label{ineq:1233}\end{align} and
  \begin{align*}
     \method{RetrieveRollingHash}&(A, S+a, H_{A,a}, i'_A + a, i'_A+d-1+a)\\&= \method{RetrieveRollingHash}(B, S, H_B, i'_A, i'_A+d-1).
  \end{align*}
  Applying Claim~\ref{claim:hash-random} over all $\tilde{O}(k)$ potential comparisons of hashes made during the algorithm, with probability at least $1 - 1/n^3$, we must have that
  \[
    \ED(A[i'_A + a, i'_A+a+d-1], B[i_B, i_B+d-1]) \le \ell.
  \]
  By (\ref{ineq:1233}), \[|i'_A + a - i_A| \le |i'_A - i_B| + |a - (i_A-i_B)| \le 2\ell + 2\ell = 4\ell.\]
  Therefore,
  \[
    \ED(A[i_A, i_A + d-1], B[i_B,i_B+d-1]) \le 9\ell.
  \]
  Therefore, $B[i_B,i_B+d-1]$ has an approximate $(i_A,\ell)$-alignment with $A$, as desired.
  \end{proof}

  \begin{theorem}\label{thm:1-preprocessing2}
  For all $A, B \in \Sigma^n$. When $A$ is preprocessed for $\tilde{O}(nk/\ell)$ time in expectation, we can distinguish between $\ED(A, B) \le \tilde{O}(k)$ and $\ED(A, B) \ge \tilde{\Omega}(k\ell)$ in $\tilde{O}(\frac{n+k^2}{\ell})$ time with probability at least $1 - \frac{1}{n}$ over the random bits in the preprocessing (oblivious to $B$).
\end{theorem}
\begin{proof}
By Lemma~\ref{lem:one-sided-shiftalign}, $\method{OneSidedMaxShiftAlign}_{\ell,k}$ is approximately correct so by Lemmas~\ref{lem:jaggedwave-pseudo1} and~\ref{lem:jaggedwave-pseudo2} succeeds with high enough probability that $\method{GreedyMatch}$ outputs the correct answer with probability at least $1 - \frac{1}{n}$.

As proved in Claim~\ref{claim:preproc-fast2}, the preprocessing runs in $\tilde{O}(nk/\ell)$ time. Also $\method{OneSidedProcessB}_{k,\ell}$ runs in $\tilde{O}(n/\ell)$ time. Further, $\method{OneSidedMaxShiftAlign}_{\ell,k}$ 
runs in $\tilde{O}(1)$ time, as performing the power-of-two search, computing the hash, and doing the table lookups are $\tilde{O}(1)$ operations. Therefore, $\method{GreedyWave}$ runs in $\tilde{O}(k^2/\ell)$ time. Thus, the whole computation takes $\tilde{O}(\frac{n+k^2}{\ell})$ time.
\end{proof}

Thus, if $\ell = k^{1-\eps}$, the proprocessing is $\tilde{O}(n \cdot k^{\eps})$, and the runtime otherwise is $\tilde{O}(n/k^{1-\eps} + k^{1+\eps}) $.

If we are in the two-sided preprocessing model, we can run both $\method{OneSidedPreprocessA}_{\ell,k}$ and $\method{OneSidedProcessB}_{\ell,k}$ for both $A$ and $B$. Then, we can just run $\method{GreedyWave}$ which runs in $\tilde{O}(k^2/\ell)$ time. Thus, we have the following corollary.
We have an implication for two-sided preprocessing as a corollary.

\begin{corollary}\label{thm:2-preprocessing2}
  For all $A, B \in \Sigma^n$. When $A$ and $B$ are preprocessed for $\tilde{O}(nk/\ell)$ time in expectation, we can distinguish between $\ED(A, B) \le \tilde{O}(k)$ and $\ED(A, B) \ge \tilde{\Omega}(k\ell)$ in $\tilde{O}(\frac{k^2}{\ell})$ time with probability at least $1 - \frac{1}{n}$ over the random bits in the preprocessing.
\end{corollary}

\section{Omitted Proofs}\label{app:omit}

\subsection{Proof of Lemma~\ref{lem:decomp}}

\begin{proof}
  We prove this by induction on $k$. If $k=0$, then we can partition $A$ and $B$ each into a single part which are matched.  Assume we have proved the theorem for $k\le k_0$. Consider $A$ and $B$ with $\ED(A, B) = k_0+1$. Thus, there exists $B' \in \Sigma^*$ such that $\ED(A, B') = k_0$ and $\ED(B', B) = 1$. By the induction hypothesis, $B'$ can be partitioned into intervals $I_1, \hdots, I_{2k_0-1}$ that are each of length at most $1$ or are equal to some interval of $A$ up to a shift of $k_0$. We now break up into cases.
  
  \textbf{Case 1.} A character of $B'$ is substituted to make $B$. Let $I_{j_0}$ be the interval this substitution occurs in. We split $I_{j_0}$ into three (some possibly empty) intervals $I_{j_0}^{(1)}, I_{j_0}^{(2)}, I_{j_0}^{3)}$, the intervals before, at, and after the substitution. We have the partition of $B$ into $2k_0+1$ intervals: $I_1, \hdots, I_{j_0-1}, I_{j_0}^{(1)}, I_{j_0}^{(2)}, I_{j_0}^{3)}, I_{j_0+1}, \hdots, I_{2k_0-1}$. Every interval is of length at most $1$ or corresponds to a equal substring of $A$ up to a shift of $k_0 \le k_0+1$. 

  \textbf{Case 2.}  A character of $B'$ is deleted to make $B$.  Let $I_{j_0}$ be the interval this deletion occurs in. We split $I_{j_0}$ into three intervals $I_{j_0}^{(1)}, I_{j_0}^{(2)}, I_{j_0}^{3)}$, the intervals before, at, and after the deletion (the middle interval is empty). We have the partition of $B$ into $2k_0+1$ intervals: $I_1, \hdots, I_{j_0-1}, I_{j_0}^{(1)}, I_{j_0}^{(2)}, I_{j_0}^{3)}, I_{j_0+1}-1, \hdots, I_{2k_0-1}-1$. Every interval is of length at most $1$ or corresponds to a equal substring of $A$ up to a shift of $k_0+1$. 

  \textbf{Case 3.} A character of $B'$ is inserted to make $B$. Let $I_{j_0}$ be the interval this insertion occurs in, or if the insertion is between intervals, take either adjacent interval. We split $I_{j_0}$ into three intervals $I_{j_0}^{(1)}, I_{j_0}^{(2)}, I_{j_0}^{(3)}$, the intervals before, at, and after the intersection. We have the partition of $B$ into $2k_0+1$ intervals: $I_1, \hdots, I_{j_0-1}, I_{j_0}^{(1)}, I_{j_0}^{(2)}, I_{j_0}^{3)}, I_{j_0+1}+1, \hdots, I_{2k_0-1}+1$. Every interval is of length at most $1$ or corresponds to a equal substring of $A$ up to a shift of $k_0+1$.

  In all three cases, if the interval $I_{j_0}$ broken up in $B$ corresponded to an interval in $A$ according to $\pi$, then the interval of $A$ can be broken up in the analogous fashion with $\pi$ suitably modified. If $I_{j_0}$ was an unmatched interval, add two empty intervals to $A$'s partition. 
\end{proof}

\end{document}